%% file: anonymous-rebels-COIN.tex
\documentclass{llncs}

\usepackage{amssymb,amsmath}
\usepackage{proof}

\usepackage{ucs}
\usepackage[utf8x]{inputenc}
\usepackage{url}

\usepackage{graphicx}
\usepackage[usenames,dvipsnames]{color}
\usepackage{tikz}
\usetikzlibrary{decorations.pathreplacing,shapes,arrows}

\usepackage[noend]{algorithmic}
\usepackage{algorithm}

\newcommand{\legalfor}[4]{{\cal C}^{#1}_{#2}(#3,#4)}

\renewcommand{\>}{\rangle}
\newcommand{\agents}{{\cal A}}

\newcommand{\likes}[2]{l(j,i)}

\newcommand{\pprofile}[2]{P(#1,#2)}

\newcommand{\strats}[3]{strat^{#1}_{#2}(#3)}

\newcommand{\votep}[1]{P(#1)}

\newcommand{\kvotep}[4]{P^{#1}_{#2}(#3,#4)}
\newcommand{\bact}[4]{{\sf Act}^{#1}_{#2}(#3,#4)}
\newcommand{\restr}[2]{{#1} \upharpoonright {#2}}

\newcommand{\act}{\mathbb{A}}

\newcommand{\noo}[1]{}

\newcommand{\dcom}[1]{\langle #1 \rangle}

\newcommand{\acro}[1]{\textsc{#1}}
\newcommand\ie{i.e{.\ }}

\newcommand{\catld}[1]{\langle\!\langle #1\rangle\!\rangle}
\newcommand{\catlb}[1]{[\![ #1 ]\!]}

\newcommand{\lnchatl}{\mathcal{L}_{\textsc{nchatl}}}

\renewcommand{\phi}{\varphi}

\title{Big, but not unruly: Tractable norms for anonymous game structures}
\author{Truls Pedersen\inst{1} \and Sjur Dyrkolbotn\inst{2} \and Piotr Kaźmierczak\inst{1,3}\thanks{Piotr Kaźmierczak's research was supported by the Research Council of Norway project 194521 (FORMGRID).}}
\institute{Dept. of Information Science and Media Studies, University of Bergen, Norway \and Durham School of Law, Durham University, UK \and Dept. of Computing, Mathematics and Physics, Bergen University College, Norway \email{truls.pedersen@infomedia.uib.no, s.k.dyrkolbotn@durham.ac.uk, phk@hib.no}}

\begin{document}

\maketitle

\begin{abstract}
  We present a new strategic logic $\acro{nchatl}$ that allows for
  reasoning about norm compliance on concurrent game structures that
  satisfy anonymity. We represent such game structures compactly,
  avoiding models that have exponential size in the number of
  agents. Then we show that model checking can be done in polynomial
  time with respect to this compact representation, even for normative
  systems that are not anonymous. That is, as long as the underlying
  game structures are anonymous, model checking normative formulas is
  tractable even if norms can prescribe different sets of forbidden
  actions to different agents.
\end{abstract}

\section{Introduction}
\label{sec:intro}

Logics of strategic ability such as Alternating-time Temporal Logic
(\acro{atl}) \cite{alur2002alternating} or Coalition Logic
\cite{Pau0202-0} have gained much interest in the multi-agent systems
community in recent years. The language of \acro{atl} (of which
Coalition Logic is the next-time fragment) allows for expressing
formulas about strategic ability of (coalitions of) agents, and it has
been used for modelling open multi-agent systems
\cite{Alur1998}. Originally, \acro{atl} was used for modelling
\emph{heterogeneous} systems. Recently, however, a semantics for
\acro{atl} tailored towards systems exhibiting some degree of
\emph{homogeneity} was presented \cite{rcgs,PedDyr13-0}.

In this paper, we continue this line of research, noting that the
\emph{homogeneity} requirement relied on in \cite{rcgs} has been
studied independently in game theory, where \emph{anonymity} is the
name given to a corresponding property of a normal form game, which is
obtained when payoff functions remain invariant under permutations of
players, see e.g., \cite{Blo00-0,Pap07-0,BraFisHol09-0}.

We tackle the question of regaining some of the expressive power lost
by requiring anonymity, and we do so by using \emph{normative
  systems}. These have emerged as a promising and powerful framework
for coordinating multi-agent systems
\cite{shoham:92a,shoham:96a,raey,AgovanRod07-0,DBLP:conf/atal/Dellunde07,AgoHoeWoo09-0}. They
allow the modeller to constrain the behaviour of agents, and can thus
provide a way to ensure that the global behaviour of the system
exhibits some desirable properties. We point out that normative
systems the way we understand them are sometimes also called
\emph{social laws}, and are simply behavioural restrictions on agents
developed by an offline designer (who is not part of the model),
much in the spirit of Shoham \& Tenneholtz's seminal
paper~\cite{shoham:96a}, and thus different from normative systems
known from deontic logic literature, since we abstract away from things
like obligations, institutions, etc. 

A key issue is the question of \emph{compliance}. Even if a normative
system is \emph{effective} in the sense that it will ensure that the
objective holds, under the assumption that all agents comply with it,
interesting questions and increased expressive power arise when one
assumes that only \emph{some} agents comply. Despite the anonymous
settings it is not irrelevant \emph{who} those agents are, in
particular, and as a consequence, a normative system provides us with
a way to regain expressive power that is lost by imposing
anonymity. We also show that doing so for anonymous game structures is
possible while maintaining the compact representation and the
tractable model checking that comes with it.

In short, our contribution in this paper combines the following four
different themes: strategic logic \acro{atl}, normative systems,
homogeneous structures and anonymous games. The resulting \emph{Norm
  Compliance Homogeneous Alternating-time Temporal Logic}
(\acro{nchatl}), in particular, arises from adding norms to
\emph{homogeneous} \acro{atl} \cite{PedDyr13-0} in a way that renders
the resulting model checking problem polynomial in the number of
agents.

The structure of the paper is as follows. In Section
\ref{sec:background} we introduce the formal background, recalling the
definition of concurrent game structures (\acro{cgs}s) and the
definition of anonymity used in game theory and social choice
theory. We also present a special case of the construction used in
\cite{rcgs}, showing how an anonymous \acro{cgs} can be succinctly
represented as a concurrent game structure with roles (an \acro{rcgs})
where the number of roles is exactly one. We go on to formulate the
notions of \emph{norms} and \emph{norm compliance} as they are used in
the multi-agent systems community. Then in Section~\ref{sec:tractable}
we define a semantics for \acro{nchatl} and investigate model checking
for this logic, showing that it is tractable. We conclude in
Section~\ref{sec:concl}.

\section{Formal Background}
\label{sec:background}

We start by introducing some definitions of the formal framework used
in the paper. The logical language we use, $\lnchatl$, is based on
\acro{atl} \cite{alur2002alternating}, extended with one extra
operator that we use to express norm compliance. Formally, the
language is generated by the following \acro{bnf}:
\[
\phi ::= \top\ |\ p\ |\ \neg \phi\ |\ \phi \lor \phi\ |\
\catld{C}\bigcirc\phi\ |\ \catld{C}\Box\phi\ |\
\catld{C}\phi\,\mathcal U\phi\ |\ \dcom{C}\phi
\]
where $p$ is a~propositional symbol, and $C$ is a coalition of agents.

The language of \acro{nchatl} contains three types of modalities:
\begin{itemize}
\item $\bigcirc$, $\Box$ and $\mathcal{U}$ are standard
  \emph{temporal} operators known from many temporal logics, and stand
  for ``next state'', ``some future state'' and ``until'',
  respectively;
\item $\catld{C}$ is a \emph{strategic ability} operator, and its
  intuitive meaning is that the coalition $\catld{C}\bigcirc\phi$ has
  a \emph{joint strategy} for enforcing a formula $\phi$ in the next
  state;
\item finally $\dcom{C}$ is the \emph{norm compliance} operator, which
  intuitive reading is that the coalition $C$ has a strategy to
  achieve $\phi$ if all its members comply to a given normative
  system.
\end{itemize}

\subsection{Anonymity}
\label{sec:anon}

We now define \emph{concurrent game structures} known from
\cite{alur2002alternating} and used for \acro{atl} interpretation, and
formalize the anonymity requirement mentioned in the introduction. We
then define the compact representation of such structures, which
provides the backbone for the semantics of \acro{nchatl}.

\begin{definition}[Concurrent Game Structure]
  \label{def:cgs}
  A \acro{cgs} is a tuple $S = \langle \agents, Q, \Pi,$ $\pi, \act,
  \delta\rangle$ where:
  \begin{itemize}
  \item $\agents$ is a non-empty set of players. In this text we
    assume $\agents = [n]$ for some $n \in \mathbb{N}$, and we reserve
    $n$ to mean the number of agents.\footnote{For the sake of
      brevity, we use the notation $[n]$ to indicate the set of
      numbers $1 \leq i \leq n$.}
  \item $Q$ is the non-empty set of states.
  \item $\Pi$ is a set of propositional letters and $\pi: Q \to
    \wp(\Pi)$ maps each state to the set of propositions true in it.
  \item $\act : Q \times \agents \to \mathbb N^+$ is the number of
    available actions in a given state. We also say that for each
    state $q \in Q$ a \emph{move vector} is a tuple $\langle
    \alpha_1,\ldots,\alpha_k\rangle$ s.t. $1 \leq \alpha_a \leq \act_a
    (q)$ for each $a \in \agents$. $D$ is then a \emph{move function}
    which given a state $q \in Q$ outputs a set of move vectors.
  \item For each $q\in Q$ and a move vector $\langle
    \alpha_1,\ldots,\alpha_k\rangle \in D(q)$ a \emph{transition
      function} produces a state $\delta(q,
    \alpha_1,\ldots,\alpha_k)\in Q$ which is a successor of $q$ when
    every agent $a \in \{1,\ldots,k\}$ chooses $\alpha_a$.
  \end{itemize}
\end{definition}

Inspired by the corresponding notion from game theory
\cite{DasPap07-0,Blo00-0,BraFisHol09-0}, we will say that a \acro{cgs}
$S$ is \emph{anonymous} if and only if:
\begin{align*}
  &  \forall_{q\in Q, i,j\in\agents}, \act(q,i) = \act(q,j) \text{ and }\\
  & \forall_{q\in Q, i,k\in\agents},
  \delta(q,\ldots,\alpha_i,\ldots,\alpha_k,\ldots) =
  \delta(q,\ldots,\alpha_k,\ldots,\alpha_i,\ldots)
\end{align*}

Any anonymous \acro{cgs} can be represented compactly as an
\acro{rcgs} -- a Concurrent Game Structure with Roles \cite{rcgs}. In
fact, the class of anonymous \acro{cgs}'s corresponds to the class of
\acro{rcgs}'s with a single role, a simplified definition of which can
be given as follows \cite{PedDyr13-0}.

\begin{definition} 
  A \acro{1rcgs} is a tuple $R = \langle \agents, Q, \Pi, \pi, \act,
  \delta \rangle$ where
  \begin{itemize}
  \item $\agents$, $Q$, $\Pi$, and $\pi$ are defined as in
    Definition~\ref{def:cgs},
  \item $\act : Q \to \mathbb N^+$ is the number of available actions
    in a given state.
  \item For every state we have a set of vectors $\votep q = \{F \in
    [n]^{[\act_q]} ~|~ \sum_{i \leq \act_q} F_i = n\}$. We will refer
    to the elements of $\votep q$ as the \emph{profiles} at $q$. For
    every state $q$ and every such profile $F \in \votep q$ we have a
    \emph{successor state} $\delta(q, F) = q'$.
  \end{itemize}
\end{definition}

The profiles assign a natural number to each action such that the sum
of these numbers (over all actions) sums up to the number of agents
$n$. The intended meaning is that the profile describes how many
agents perform each action. We also define partial profiles at $q \in
Q$, for all $A \subseteq \agents$ as follows:
$$
\pprofile q A = \left\{~F \in [n]^{[\act_q]}\ ~\left|~ \sum_{i \leq
      \act_q} F_i = |A|\right.~\right\}.
$$

It is not hard to see that a \acro{1rcgs} can be given to provide a
succinct representation for any \acro{cgs} which satisfies the
anonymity requirement; as the permutations of the action profiles are
irrelevant, we only need to record \emph{how many} agents performed
each action.

\subsection{Normative systems}
\label{sec:norms}

Following \cite{AgoHoeWoo09-0,raey} we define a normative system as a
map $\eta: Q \times \agents \to 2^{\mathbb N^+}$, giving, for each
state and agent, the set of actions that are forbidden for that agent
in that state. We require $\eta(q,a) \in [\act_q]$ for all $q \in Q,a
\in \agents$ and that $\eta$ is such that, for every state, there is
at least \emph{some} legal action. That is, $\forall q \in Q, a \in
\agents: [ \act_q(a) ] \setminus \eta (q,a) \neq \emptyset$.

To account for ``disobedience'' of certain agents (\ie those that do
\emph{not} comply with a normative system), we consider normative
systems \emph{restricted} to specific coalitions. Such a restriction
means that only actions that are controlled by a given coalition are
blacklisted (intuitively, anyone not belonging to that coalition is
free not to comply with the normative system). We use the
$\upharpoonright$ symbol to denote such restrictions, and formally
define it below:

\[
(\eta \upharpoonright C) (q, a) = \begin{cases}
  \eta(q, a) & \text{if}\ a\in C\\
  \emptyset & \text{otherwise.}
\end{cases}
\]

Notice that normative systems are not anonymous. It would certainly
also be possible to consider anonymous norms, \ie norms that are
invariant under agents' names and simply forbid actions at
states. However, since our main result is that non-anonymous norms are
tractable on anonymous structures, we will not pay much attention to
this special case in this paper. We show however, that for the example
from Section \ref{sec:ex}, moving from anonymous to non-anonymous
norms gives us increased expressive power.

\section{Tractable norms for anonymous game structures}
\label{sec:tractable}

The semantic structures we will use in this section are defined as
follows, following \cite{rcgs}.
\begin{definition}
  A normative \acro{1rcgs} is a pair $H = \langle R, \eta \rangle$
  where:
  \begin{itemize}
  \item $R$ is a \acro{1rcgs}, and
  \item $\eta$ is a normative system for $R$.
  \end{itemize}
\end{definition}

In addition to the notions introduced for \acro{1rcgs}s, we also need
access to the partial profiles which the agents in some coalition $B$
can choose, assuming that agents in some other coalition $A$ comply to
$\eta$. The straightforward way of defining such profiles is to go via
an explicit representation of compliant action-tuples for $B$, defined
as follows.

\begin{definition}\label{def:bact}
  Given a \acro{1rcgs} $H$, a state $q$ in $H$ and two coalitions $A,B
  \subseteq \agents$, an $\restr \eta A$-compatible $B$-action at $q$
  is a vector $\rho : B \to \mathbb N^+$ such that:
$$\forall b \in B ~:~ \rho(b) \in [\act_q] \setminus (\restr \eta A)(q, b).$$
We let $\bact \eta A q B$ denote the set of all $\restr \eta
A$-compatible $B$-actions at $q$.
\end{definition}

Then $B$-actions give rise to $B$-profiles as follows.

\begin{definition}
  \label{def:cbp}
  If $\rho$ is an $\restr \eta A$-compatible $B$-action, then the
  corresponding $\restr \eta A$-compatible $B$-profile is a vector
  $s_B^{\restr \eta A} : [\act_q] \to \mathbb N$ such that:
$$\forall i \in [\act_q]:\left(s_B^{\restr \eta A}\right)(i) = \left|\{ b \in B ~|~ \rho(b) = i\}\right|. $$
We gather all $\restr \eta A$-compatible $B$-profiles for which there
is a corresponding $\restr \eta A$-compatible $B$-action at $q$ in the
set $\kvotep \eta A q B$.
\end{definition}

Notice that a direct computation of this set, using Definition
\ref{def:bact}, requires computing the set $\bact \eta A q B$, which
can have exponential size in the number of agents from $B$. This would
defeat the purpose of compact representation, the aim of which is to
ensure that complexity of model checking remains polynomial in the
number of agents as long as the number of actions is constant. It
turns out, however, that computation of $\bact \eta A q B$ can be
avoided for \emph{arbitrary} (non-anonymous) normative systems, and
that a polynomial-time procedure can be used instead. We return to
this challenge in Section~\ref{char}, after we have defined truth on
normative \acro{1rcgs} models.

To do this, we need some more notation. Given $F \in \kvotep \eta A q
B, G \in \kvotep \eta C q D$, we say that $F \geq G$ if $A = C$ and
for every $i \in [\act_q]$ we have $F_i \geq G_i$. Given two states
$q,q' \in Q$, we say that $q'$ is a \emph{successor} of $q$ if there
is some $F \in \votep q$ such that $\delta(q,F) = q'$. A
\emph{computation} is an infinite sequence $\lambda = q_0q_1\ldots$ of
states such that for all positions $i\geq 0$, $q_{i+1}$ is a successor
of $q_i$. We follow standard abbreviations, hence a $q$-computation
denotes a computation starting at $q$, and $\lambda[i]$,
$\lambda[0,i]$ and $\lambda[i,\infty]$ denote the $i$-th state, the
finite prefix $q_0q_1\ldots q_i$ and the infinite suffix
$q_iq_{i+1}\ldots$ of $\lambda$ for any computation
$\lambda$ and its position $i\geq 0$, respectively. \\

\begin{definition} An $\restr \eta A$-compatible $B$-strategy is a map
  $s_B: Q \to \bigcup_{q \in Q} \kvotep \eta A q B$ such that:
$$s_B(q) \in \kvotep \eta A q B \text{ for each } q \in Q.$$
We denote the set of all such strategies by $\strats \eta A B$.
\end{definition}

Notice that if $s \in \strats \eta A \agents$ for some $A \subseteq
\agents$, then if we apply $\delta(q)$ to $s(q)$ we obtain a unique
new state $q' = \delta(q, s(q))$. Iterating, we get the \emph{induced}
computation $\lambda_{s,q} = q_{0}q_{1}\ldots$ such that $q = q_0$ and
$\forall i~\geq 0: \delta(q_{i}, (s(q_i))) = q_{i+1}$. Given $s_B \in
\strats \eta A B$ and a state $q$ we get an associated \emph{set} of
computations $out(s_B,q)$. This is the set of all computations that
can result when at any state, $B$ is acting in the way specified by
$s_B$. That is,
\begin{equation}\label{eq:nout}
  out(s_B,q) ~:=~ \{\lambda_{s,q} \mid s \in \strats \eta A \agents \text{ and } s_B \leq s\}.
\end{equation}

We can now define normative satisfaction on \acro{1rcgs}'s as follows.
\begin{definition}\label{def:nhatl}
  Given a normative \acro{1rcgs} $(H, \eta)$, a state $q$ and a
  coalition $A \subseteq \agents$, truth of $\phi$ on $(H, \eta)$
  under $A$-compliance is defined inductively.
  \begin{itemize}
  \item $H, \eta, A, q \models p$ iff $q \in \pi(p)$
  \item $H, \eta, A, q \models \neg\phi$ iff $H, \eta, A, q
    \not\models \phi$
  \item $H, \eta, A, q \models \phi\lor\psi$ iff $H, \eta, A, q
    \models \phi$ or $H, \eta, A, q \models \psi$
  \item $H, \eta, A, q \models \catld{C}\!\bigcirc\!\phi$ iff $\exists
    s_C \in \strats \eta A C: \forall \lambda \in out(s_C, q):
    \lambda[1] \models \phi$
  \item $H, \eta, A, q \models \catld{C}\Box\phi$ iff $\exists s_C \in
    \strats \eta A C: \forall \lambda \in out(s_C, q):$\\ $\forall
    i\geq 0: \lambda[i] \models \phi$
  \item $H, \eta, A, q \models \catld{C}\phi\mathcal{U}\psi$ iff
    $\exists s_C \in \strats \eta A C: \forall \lambda \in out(s_C,
    q): \exists i \geq 0: (\lambda[i] \models \psi \land \forall j \in
    [i]: \lambda[j]\models \phi)$
  \item $H, \eta, A, q \models \dcom{B}\phi$ iff $H, \eta, B, q
    \models \phi$
  \end{itemize}
\end{definition}

Clearly, to solve the model checking problem for this logic, we need
to compute sets of the form $\kvotep \eta A q B$, and how to do this
efficiently is the main obstacle preventing a quick algorithm. We
address and resolve this challenge in Section~\ref{char}, but first we
consider an example.

\subsection{Example}\label{sec:ex}

For a simple illustration of the kind of reasoning we can perform
using norms on anonymous game structures let us assume we have a
system set up to perform two tasks, $p_1$ and $p_2$. Let us further
assume that the system contains agents $\agents = [n]$ where, for
simplicity, we assume $n$ is a multiple of $10$. Also, assume that
every agent must choose to contribute to either $p_1$ or $p_2$, a
choice we encode as a choice between shared actions $\alpha_{p_1}$ and
$\alpha_{p_2}$. If the task $p_1$ is successfully performed, $p_1$
becomes true in the next state, and similarly for the task $p_2$.

As it happens, our system is such that in order for $p_1$ to be
successfully performed we need $80 - 90 \%$ of the agents to
contribute towards $p_1$. That is, for $p_1$ to become true, such a
percentage of agents have to choose $\alpha_{p_1}$ as their action. On
the other hand, in order for $p_2$ to be true in the next state, we
need $20 - 60 \%$ of the agents to perform $\alpha_{p_2}$. In
Figure~\ref{fig:example-1} we depict an \acro{1rcgs} modelling such a
scenario.\footnote{The pairs used to decorate transitions denote
  profiles, with the first coordinate being the percentage of agents
  doing $\alpha_{p_1}$, and the second coordinate being those who do
  $\alpha_{p_2}$. We have $\delta(q_0,\<i,j\>) = q_{i,j}$ for all such
  tuples. We omit reflexive loops for all states $q_{i,j}$.}

Notice that if \emph{both} $p_1$ and $p_2$ are to be performed
successfully, we need precisely $20 \%$ of the agents to perform
$\alpha_{p_2}$ while the remaining $80 \%$ choose to do
$\alpha_{p_1}$. It follows that in order to successfully complete both
tasks, we need \emph{coordination}. In fact, as it stands, we need
everyone to coordinate their actions with everyone else. In terms of
\acro{atl}, since successful completion of both $p_1$ and $p_2$
results from a \emph{unique} profile, only the \emph{grand coalition}
can ensure $p \land q$. That is, while we have $H,q \models
\catld{\agents}\bigcirc (p \land q)$, we also have $H,q \models
\catlb{A}\bigcirc (\neg p \lor \neg q)$ for all $A \subset
\agents$.\footnote{$\catlb{}$ is the dual of the strategic ability
  operator $\catld{}$. Intuitively, $\catlb{A}\phi$ means that
  coalition $A$ can not avoid $\phi$. }

Moreover, notice that even if some coalition $A \subseteq \agents$ can
observe what the agents in $\agents \setminus A$ do, they might not
necessarily respond in such a way that $p \land q$ becomes true. To
see this, assume that $A$ contains $60 \%$ of the agents. Then if the
remaining agents all perform $\alpha_{p_2}$, it becomes impossible for
$A$ to respond in such a way that $p_1$ becomes true. We have, in
particular, $\catld{\agents \setminus A}\bigcirc \neg p$.

\begin{figure}[ht]

  \centering

  \begin{tikzpicture}[every node/.style={circle, draw, scale=1},
    scale=0.8, rotate=180]

    \tikzset{every node/.style={scale=1, thick}}

    \node (q0) at (0.0, -2.0) {$\begin{smallmatrix}q \\
        \emptyset\end{smallmatrix}$};

    \node (q0-100) at (4.5, 3.0) {$\begin{smallmatrix}q_{0,100} \\
        \emptyset\end{smallmatrix}$};

    \node (q10-90) at (3.3, 3.0) {$\begin{smallmatrix}q_{10,90} \\
        \{p_1\}\end{smallmatrix}$};

    \node (q20-80) at (2.1, 3.0) {$\begin{smallmatrix}q_{20,80} \\
        \{p_1,p_2\}\end{smallmatrix}$};

    \node (ddd1) at (0.4, 3.0) {$\dots$};

    \node (q60-40) at (-1.2, 3.0) {$\begin{smallmatrix}q_{60,40} \\
        \{p_2\}\end{smallmatrix}$};

    \node (ddd2) at (-2.9, 3.0) {$\dots$};

    \node (q100-0) at (-4.5, 3.0) {$\begin{smallmatrix}q_{100,0} \\
        \emptyset\end{smallmatrix}$};


    \tikzset{every node/.style={solid}};

    \tikzset{mystyle/.style={->,thick}};

    \path (q0) edge [mystyle, left] node {$\begin{smallmatrix}\langle
        0,100 \rangle\end{smallmatrix}$} (q0-100);

    \path (q0) edge [mystyle] node {} (q10-90);

    \path (q0) edge [mystyle, right] node {$\begin{smallmatrix}\langle
        20,80 \rangle\end{smallmatrix}$} (q20-80);

    \path (q0) edge [mystyle, below, right] node
    {$\begin{smallmatrix}\langle 60,40 \rangle\end{smallmatrix}$}
    (q60-40);

    \path (q0) edge [mystyle, right] node {$\begin{smallmatrix}\langle
        100,0 \rangle\end{smallmatrix}$} (q100-0);

    \draw [decorate,decoration={brace,amplitude=5pt},thick](3.5,3.4)
    -- (4.5,3.4) node [black,midway,yshift=-5pt,below] {\footnotesize
      $\emptyset$};

    \draw [decorate,decoration={brace,amplitude=5pt},thick](1.8,3.4)
    -- (3.5,3.4) node [black,midway,yshift=-5pt,below] {\footnotesize
      $p_1$};

    \draw [decorate,decoration={brace,amplitude=5pt},thick](-1.4,3.4)
    -- (2.4,3.4) node [black,midway,yshift=-5pt,below] {\footnotesize
      $p_2$};

    \draw [decorate,decoration={brace,amplitude=5pt},thick](-4.5,3.4)
    -- (-1.4,3.4) node [black,midway,yshift=-5pt,below] {\footnotesize
      $\emptyset$};

  \end{tikzpicture}

  \caption{A coordination problem resolved by norms}

  \label{fig:example-1}

\end{figure}
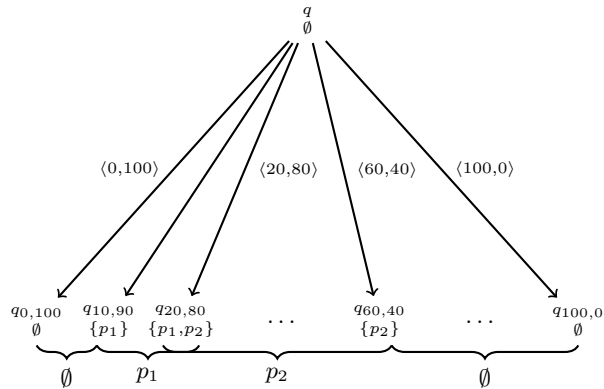

Suppose that we want to use norms to achieve $p \land q$ even under
the assumption that only those agents that are in $A$ are capable of
coordinating their actions. Clearly, this is possible. For instance,
if we simply demand that $\agents \setminus A$ all perform the same
action, and they comply, then, assuming the norm to be common
knowledge, $A$ can adapt accordingly. Somewhat more subtly, notice
that in order to ensure $\catlb{\agents \setminus A}\bigcirc(p\land
q)$ we do not require such a powerful norm. It is sufficient, in
particular, to fix some $B \subseteq \agents \setminus A$ containing
$20 \%$ of the agents, and introduce the norm $\eta$ defined by:
$$
\eta(q,a) = \begin{cases} \{\alpha_{p_2}\} \text{ if } q=q_0, a \in B
  \\ \emptyset \text{ otherwise. } \end{cases}
$$
As long as $B$ complies, $A$ can indeed achieve $p \land q$ \emph{as
  long as they observe} what the other agents do and adapt
accordingly.  In logical terms, we have $H,\eta,\emptyset,q_0 \models
\dcom B \catlb{\agents \setminus A}\bigcirc (p \land q)$. If it is not
obvious, we leave it to the reader to verify this, possibly by using
\texttt{mcheck} from Algorithm~\ref{alg:mcheck}.

The toy example considered here also serves to illustrate that
non-anonymous norms give increased expressive power compared to norms
that just forbid a set of actions. Consider, in particular, the
situation when we want to empower $A$ to \emph{choose} whether $p_1$
or $p_2$ is to become true, irrespectively of what the remaining
agents do. Using a non-anonymous norm, this can be achieved by
choosing $B' \in \agents$ containing $10 \%$ of the agents such that
$B' \cap A = B' \cap B = \emptyset$. To see this, consider the norm
$\eta'$ defined by:
$$
\eta'(q,a) = \begin{cases} \{\alpha_{p_2}\} \text{ if } q=q_0, a \in B
  \\ \{\alpha_{p_1}\} \text{ if } q = q_0,a \in B' \\ \emptyset \text{
    otherwise. } \end{cases}
$$
Then, as long as $B \cup B'$ comply, we have at least $10 \%$ doing
$\alpha_{p_2}$ and $20 \%$ doing $\alpha_{p_1}$, from which it follows
that $p_1$ is ensured as long as all members of $A$ perform
$\alpha_{p_1}$, while $p_2$ is ensured, for instance, if $50 \%$ of
the members in $A$ perform $\alpha_{p_2}$. We have, in particular,
$H,\eta',\emptyset,q \models \dcom {B \cup B'}( \catld A \bigcirc p
\land \catld A \bigcirc q)$. It is not hard to see that no anonymous
norm can achieve this, as long as only $30 \%$ of the agents are
assumed to comply with it.

\subsection{Characterizing $\restr \eta A$-compatible
  $B$-profiles}\label{char}

In this section, we will provide a characterization showing that quick
computation of the sets $\kvotep \eta A q B$ is indeed
possible. Towards this result, we first observe the following simple
fact, the proof of which is trivial and omitted.

Whenever we use the ``$+$'' symbol with respect to vectors, we mean
addition coordinate-wise.

\begin{proposition}\label{prop:plus}
  Given a \acro{1rcgs}, a normative system $\eta$, coalitions $A,B
  \subseteq \agents$ and a state $q \in Q$, we have $F \in \kvotep
  \eta A q B$ if, and only if,
$$\exists F_1 \in \pprofile q {B \setminus A} ,~\exists F_2 \in \kvotep \eta A q {A \cap B} ~\text{s.t.}~ F = F_1 + F_2.$$
\end{proposition}

We will also need the following auxiliary function.

\begin{definition}\label{def:legalfor}
  Given an \acro{1rcgs} $H$, a normative system $\eta$ and any state
  $q \in Q$ we define, for all $E \subseteq \act_q,A \subseteq
  \agents$, the following set:
$$
\legalfor q \eta E A = |\{x \in A \mid \eta(q,x) \cap E \not =
\emptyset\}|.
$$ 
\end{definition}
So $\legalfor q \eta E A$ returns the \emph{number} of agents in $A$
that have a legal action in $E$ at $q$. Using this function allows us
to characterize $\kvotep \eta A q B$ more compactly using a matching
argument, giving rise to the following lemma, towards tractable model
checking.

\begin{lemma}\label{lemma:main}
  For any \acro{1rcgs}, any normative system $\eta$ and any $A,B
  \subseteq \agents$ we have $F \in \kvotep \eta A q B$ iff $F = F_1 +
  F_2$ for some $F_1 \in \pprofile q {B \setminus A}$ and some $F_2
  \in \pprofile q {A \cap B}$ such that:
  \begin{equation}\label{eq:allsubset}
    \forall E \subseteq [\act_q]: \legalfor q \eta E {A \cap B} \geq \sum_{i \in E}F_2(i).
  \end{equation}
\end{lemma}

\begin{proof}
  $\Rightarrow$) Trivial. \\
  $\Leftarrow$) Assume that we have $F = F_1 + F_2$ for $F_1 \in
  \pprofile q {B \setminus A}$ and $F_2 \in \pprofile q {A \cap B}$
  such that (\ref{eq:allsubset}) holds. We demonstrate existence of
  $\rho \in \bact \eta A q B$ that induces the profile $F_2$, i.e.,
  such that
$$F_2(i) = |\{ x \in A \cap B ~|~ \rho(x) = i\}| \text{ for all } i \in [\act_q].$$
We will think of $\rho$ as the solution of a matching problem in a
bipartite graph: Let $G = (V_1,V_2,E)$ where $V_1 = A \cap B$, $V_2 =
\{i_j \mid i \in [\act_q], j \in [F_2(i)]\}$ are the two sets of nodes
and $E = \{(x,i_j) \mid i \not \in \eta(q,x)\}$ is the set of edges.
Notice that $|V_1| = |V_2|$ since $F_2 \in \pprofile q {A \cap B}$,
and that the graph is indeed bipartite. For all subsets of $V
\subseteq V_2$, let $V^- = \{x \in V_1 \mid \exists i_j \in V_2:
(x,i_j) \in E\}$. Then, since $F_2$ satisfies (\ref{eq:allsubset}), it
follows that for all $V \subseteq V_2$ we have $|V^-| \geq V$. This
means that the conditions of Hall's marriage theorem are all fulfilled
(well known from graph theory, originally published in
\cite{Hal35-0}), meaning that there exists a set $E' \subseteq E$ such
that for every $i_j \in V_2$ there is a unique $x \in V_1$ such that
$(x, i_j) \in E'$, i.e., such that $E'$ is a matching in $G$. Let us
define the vector $\rho: A \cap B \to \mathbb N^+$ such that $\rho(x)
= i$ for all $(x,i_j) \in E'$. Clearly, since $E'$ is a matching, this
is well-defined and we have $\rho \in \bact \eta A q B$ as
desired. Moreover, it is easy to see that $\rho$ corresponds to $F_2$
in the sense of Definition~\ref{def:cbp}. We conclude that $F_2 \in
\kvotep \eta A q {B \cap A}$. Then, from Proposition~\ref{prop:plus}
it follows that $F = F_1 + F_2 \in \kvotep \eta A q B$, concluding the
proof.\qed
\end{proof}

In Figure~\ref{fig:lemma1} we illustrate how $\kvotep \eta A q B$ is
generated, by calculating the sets $\pprofile q {B \setminus A}$ and
$\pprofile q {B \cap A}$ the latter of which is then restricted to the
elements which satisfy Condition~(\ref{eq:allsubset}) (in
Lemma~\ref{lemma:main}). The parameters of the situation illustrated
is the number of actions $\act_q = 3$ and the set of agents $\agents =
\{a, b, c, d, e\}$ with agents $A = \{b, c, d\}$ complying to $\eta$,
and the agents for which we are making a set of profiles for are
contained in $B = \{c, d, e\}$. The consequence of $\eta$ for the two
agents in $B$ which do comply, is that it forbids action $2$ for agent
$c$ and actions $1$ and $2$ for agent $d$.

\begin{figure}[h]
  \def\svgwidth{\textwidth} 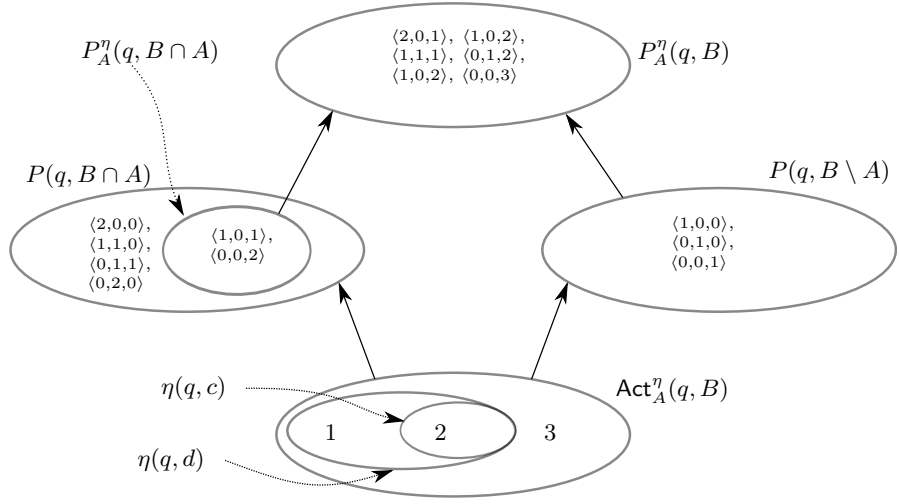
  \caption{Illustration of Lemma 1.}
  \label{fig:lemma1}
\end{figure}

In light of Lemma~\ref{lemma:main}, it is clearly possible, as long as
the number of actions is constant, to generate $\kvotep \eta A q B$ in
polynomial time for all $A,B,q$. We simply run through all $F \in
\pprofile q {A \cap B}$ and check if Condition~(\ref{eq:allsubset})
holds. This involves running though all subsets of $\act_q$, but still
it only requires a constant number of traversals of $A \cap B$. Then
the set $\kvotep \eta A q B$ is obtained from any such $F$ passing the
test, when added to any vector from the set $\pprofile q {B \setminus
  A}$, as detailed in Algorithm~\ref{alg:kvotep-full}.

We mention that the construction in the proof of
Lemma~\ref{lemma:main} mirrors the construction used in
\cite{BraFisHol09-0} to establish that finding pure Nash equilibria in
an anonymous normal form game is decidable in polynomial time provided
the number of actions remain constant. This result, in particular, is
also obtained by an application of Hall's marriage theorem.

More importantly, given an \acro{1rcgs} $H$, a normative system
$\eta$, a state $q \in Q$ and coalitions $A, B$, it seems clear that
we can define an anonymous normal form game such that $\kvotep \eta A
q B$ is the set of pure Nash equilibria in this game. We omit the
details due to space restrictions, but remark that as
Lemma~\ref{lemma:main} can be seen as a corollary of results from
\cite{BraFisHol09-0}, it follows that computing $\kvotep \eta A q B$
can also be done by employing the more subtle techniques introduced
there, used to prove membership in the complexity class $\sf
TC^0$. This means, in particular, that the algorithm presented in the
next section, while showing that model checking is tractable, could be
improved on this point. Here, however, we do not focus on the design
of optimal procedures, but on clearly conveying the main result and
the ideas that have precipitated it.


\subsection{Tractable model checking}\label{tract}

The algorithm for checking truth of $\phi$ in a normative \acro{1rcgs}
follows exactly the same pattern as the standard model checking
algorithm used to do model checking on \acro{cgs} models, see e.g.,
\cite{Jamroga:2009:EYH:1615285.1615287}. Given a \acro{cgs} model $S$
and a formula $\phi$, this algorithm processes $\phi$ recursively and
returns the set of states $q \in Q$ where $\phi$ is true. To deal
correctly with $\catld{A}\Box\phi$ and $\catld{A}\phi\mathcal{U}\psi$
the algorithm relies on the following fixed point characterizations,
which are well-known to hold for \acro{atl}, see for instance
\cite{Jamroga:2009:EYH:1615285.1615287}, and are also easily seen to
be true on any normative \acro{1rcgs} model, c.f.,
Definition~\ref{def:nhatl}:
\begin{gather}\label{eq:fixed}
  \begin{gathered}
    \catld{A}\Box\phi \leftrightarrow \phi \land
    \catld{A}\bigcirc\catld{A}\Box\phi\\
    \catld{A}\phi_{1}\mathcal{U}\phi_{2} \leftrightarrow \phi_{2}\lor
    (\phi_{1}\land\catld{A}\bigcirc\catld{A}\phi_{1}\mathcal{U}\phi_{2}
  \end{gathered}
\end{gather}

In light of this, the correctness of the algorithm \texttt{mcheck},
shown in Algorithm~\ref{alg:mcheck}, follows trivially if we can
establish correctness of the algorithm \texttt{enforce}, shown in
Algorithm~\ref{alg:enforce}.

\begin{algorithm}[h]
  \begin{algorithmic}
    \IF{$\phi = p \in \Pi$} \RETURN $\pi(p)$
    \ENDIF
    \IF {$\phi = \neg \psi$} \RETURN $Q \setminus
    \mathtt{mcheck}(H,\eta,A,\psi)$
    \ENDIF
    \IF {$\phi = \psi \vee \psi'$} \RETURN
    $\mathtt{mcheck}(H,\eta,A,\psi) \cup
    \mathtt{mcheck}(H,\eta,A,\psi')$
    \ENDIF
    \IF {$\phi = \catld{B}\bigcirc\psi$} \RETURN $\{q \mid
    \mathtt{enforce}(H,\eta,A,q,B, \mathtt{mcheck}(H,\eta,A,\psi))\}$
    \ENDIF
    \IF {$\phi = \catld{B}\Box\psi$} \STATE $Q_1 := Q$ \STATE $Q_2 :=
    \mathtt{mcheck}(H,\eta,A,\psi)$ \WHILE {$Q_1 \not \subseteq Q_2$}
    \STATE $Q_1 := Q_2$ \STATE $Q_2 := \{q \in Q \mid
    \mathtt{enforce}(H,\eta,A,B,q,Q_2)\} \cap Q_2$
    \ENDWHILE
    \RETURN $Q_1$
    \ENDIF
    \IF {$\phi = \catld{B}\psi\mathcal{U}\psi'$} \STATE $Q_1 :=
    \emptyset$ \STATE $Q_2 = mcheck(H,\psi)$ \STATE $Q_3 =
    mcheck(H,\psi')$ \WHILE{$Q_3 \not \subseteq Q_1$} \STATE $Q_1 :=
    Q_1 \cup Q_3$ \STATE $Q_3 := \{q \in Q \mid
    \mathtt{enforce}(H,\eta,A,B,q,Q_1)\} \cap Q_2$
    \ENDWHILE
    \RETURN $Q_3$
    \ENDIF
    \IF {$\phi = \dcom{A'}\psi$} \RETURN
    $\mathtt{mcheck}(H,\eta,A',\psi)$
    \ENDIF
  \end{algorithmic}
  \caption{\texttt{mcheck}$(H,\eta, A,\phi)$ algorithm}
  \label{alg:mcheck}
\end{algorithm}

\begin{algorithm}[h]
  \begin{algorithmic}
    \STATE $S_{pro} = \mathtt{comp}(\eta,A,q,B)$ // $S_{pro} = \kvotep
    \eta A q B$ \STATE $S_{ant} = \mathtt{comp}(\eta,A,q,B)$ //
    $S_{ant} = \kvotep \eta A q {(\agents \setminus B)}$ \FOR {$F_B
      \in S_{pro}$} \STATE $x = true$ \FOR {$F_{B'} \in S_{ant}$} \IF
    {$\delta_q(F_B + F_{B'}) \notin Q'$} \STATE $x = false$
    \ENDIF
    \ENDFOR
    \IF {$x = true$} \RETURN true
    \ENDIF
    \ENDFOR
    \RETURN false
  \end{algorithmic}
  \caption{\texttt{enforce}$(H, \eta,A q, B, Q')$ algorithm}
  \label{alg:enforce}
\end{algorithm}

This algorithm answers, given a normative \acro{1rcgs} $(H,\eta)$, a
state $q$, coalitions $A,B \subseteq \agents$ and a set of states
$Q'$, whether or not there is some strategy $s_B \in \strats \eta A B$
such that $\{\lambda[1] \in Q ~|~ \lambda \in out(s_B, q)\} \subseteq
Q'$. Clearly, such a strategy exists if, and only if, there is some
$F_B \in \kvotep \eta A q B$ such that for all $F \in \kvotep \eta A q
\agents$, if $F_A \leq F$ then $\delta(q,F) \in Q'$. Thus, correctness
of \texttt{enforce} follows if the algorithm \texttt{comp}, shown in
Algorithm~\ref{alg:kvotep-full}, correctly computes the necessary sets
$\kvotep \eta A q B$. This, in turn, clearly follows from
Lemma~\ref{lemma:main}. To see this, notice that the step when we
place agents in the set $T$ corresponds exactly to the calculation of
$\legalfor q \eta E A$.\footnote{This implementation of the $\kvotep
  \eta A q B$, collecting agents in $T$, could be optimized if we just
  count the first occurrence of a satisfying condition (where $a$ is
  added to $T$) and move on to the next agent. We use a set to
  simplify the presentation.}

Moreover, notice that all of the procedures involved in model checking
have polynomial complexity in the length of the formula and the size
of the model. This follows by the fact that the sizes of $S_{pro}$ and
$S_{ant}$, used by \texttt{enforce} and calculated by \texttt{comp},
have sizes bounded above by $\frac{(|B| + (|\act_q| -
  1))!}{|B|!(|\act_q| -1)!}$ and $\frac{(|\agents \setminus B| +
  (|\act_q| - 1))!}{|\agents \setminus B|!(|\act_q| -1)!}$
respectively.  These combinatorial expressions are both bounded above
by $|\agents|^{|\act_q|}$, so there is indeed no exponential
dependence on the number of agents, only on the number of
actions. Also remember that we compute $S_{pro}$ and $S_{ant}$
effectively, by applying Lemma~\ref{lemma:main}. The main result
follows.

\begin{theorem}\label{thm:main}
  Given a normative \acro{1rcgs} $(H,\eta)$, a state $q \in Q$, a
  coalition $A \subseteq \agents$ and a formula $\phi$: Deciding if
  $H,\eta,A,q \models \phi$ takes polynomial time in the size of $H$
  and the length of $\phi$.
\end{theorem}

\begin{algorithm}[h]
  \begin{algorithmic}
    \STATE $R = \pprofile q {B\setminus A}$ \FOR {$x \in \pprofile q
      {B \cap A}$} \STATE $y = true$ \FOR {$E \subseteq \act_q$}
    \STATE $T = \emptyset$ \FOR {$a \in B \cap A$} \FOR {$\alpha \in
      E$} \IF {$\alpha \notin \eta(q, a)$} \STATE {add $a$ to $T$}
    \ENDIF
    \ENDFOR
    \ENDFOR
    \ENDFOR
    \IF {$|T| < \sum_{\alpha \in E} x_\alpha$ } \STATE $y = false$
    \ENDIF
    \IF {$y = true$} \STATE add $x$ to $R$
    \ENDIF
    \ENDFOR
    \RETURN $R$
  \end{algorithmic}
  \caption{\texttt{comp}$(\eta,A,q,B)$ algorithm including $T =
    \legalfor q \eta E A$}
  \label{alg:kvotep-full}
\end{algorithm}

\section{Conclusion}
\label{sec:concl}

In this paper, we have considered concurrent game structures that
satisfy anonymity. Following \cite{rcgs,PedDyr13-0}, we represent
these structures compactly, avoiding models that have exponential size
in the number of agents. Then we consider normative systems applied to
such models, resulting in the logic \acro{nchatl}. Our main technical
result is that this logic still admits a tractable algorithm for the
model checking problem.

More generally, we believe our work serves to establish interesting
connections, both conceptual and technical, between recent work in
algorithmic game theory and recent work on logics for strategic
ability of coalitions of agents. It seems, in particular, that a major
challenge which is becoming increasingly important to both these
fields is the need for compact representations, allowing us to make
use of established formalisms to analyse systems with a large number
of participating agents.

In order for this to become feasible in practice, we certainly require
representations and notions that avoid introducing exponential
time-dependence on the number of agents that are present. The danger,
however, is that when formulating restrictions that make this
possible, one deprives the underlying formalism of crucial expressive
power. In this paper, we have addressed this worry for \acro{atl}, and
shown that norms can be used to regain some of what is lost by
requiring anonymity.

Moreover, and somewhat surprisingly, it turns out that even
non-homogeneous norms can be implemented without introducing any
exponential dependence on the agents. We conclude, therefore, that
normative systems are a good candidate in general for giving compact
multi-agent formalism a limited, but useful, means for talking about
such heterogeneous properties that can be expressed without resulting
in an exponential blow-up of crucial decision problems.

\end{document}

%% file: lemma1.pdf_tex
\begingroup%
  \makeatletter%
  \providecommand\color[2][]{%
    \errmessage{(Inkscape) Color is used for the text in Inkscape, but the package 'color.sty' is not loaded}%
    \renewcommand\color[2][]{}%
  }%
  \providecommand\transparent[1]{%
    \errmessage{(Inkscape) Transparency is used (non-zero) for the text in Inkscape, but the package 'transparent.sty' is not loaded}%
    \renewcommand\transparent[1]{}%
  }%
  \providecommand\rotatebox[2]{#2}%
  \ifx\svgwidth\undefined%
    \setlength{\unitlength}{584.25bp}%
    \ifx\svgscale\undefined%
      \relax%
    \else%
      \setlength{\unitlength}{\unitlength * \real{\svgscale}}%
    \fi%
  \else%
    \setlength{\unitlength}{\svgwidth}%
  \fi%
  \global\let\svgwidth\undefined%
  \global\let\svgscale\undefined%
  \makeatother%
  \begin{picture}(1,0.6182285)%
    \put(0,0){\includegraphics[width=\unitlength]{lemma1.pdf}}%
    \put(0.09543931,0.27940201){\color[rgb]{0,0,0}\makebox(0,0)[lb]{\smash{$\begin{smallmatrix}
\langle 2, 0, 0\rangle,\\
\langle 1, 1, 0\rangle,\\
\langle 0, 1, 1\rangle,\\
\langle 0, 2, 0\rangle~
\end{smallmatrix}$}}}%
    \put(0.22822447,0.28908431){\color[rgb]{0,0,0}\makebox(0,0)[lb]{\smash{$\begin{smallmatrix}
\langle 1, 0, 1\rangle,\\
\langle 0, 0, 2\rangle~
\end{smallmatrix}$}}}%
    \put(0.35871535,0.08415102){\color[rgb]{0,0,0}\makebox(0,0)[lb]{\smash{$1$}}}%
    \put(0.5977992,0.08415102){\color[rgb]{0,0,0}\makebox(0,0)[lb]{\smash{$3$}}}%
    \put(0.4782573,0.08415102){\color[rgb]{0,0,0}\makebox(0,0)[lb]{\smash{$2$}}}%
    \put(0.1804139,0.13448256){\color[rgb]{0,0,0}\makebox(0,0)[lb]{\smash{$\eta(q,c)$}}}%
    \put(0.15448045,0.05530687){\color[rgb]{0,0,0}\makebox(0,0)[lb]{\smash{$\eta(q,d)$}}}%
    \put(0.67572049,0.13147423){\color[rgb]{0,0,0}\makebox(0,0)[lb]{\smash{$\bact \eta A q B$}}}%
    \put(0.69931595,0.50005002){\color[rgb]{0,0,0}\makebox(0,0)[lb]{\smash{$\kvotep \eta A q B$}}}%
    \put(0.09306819,0.50037684){\color[rgb]{0,0,0}\makebox(0,0)[lb]{\smash{$\kvotep \eta A q {B \cap A}$}}}%
    \put(0.8421997,0.37006203){\color[rgb]{0,0,0}\makebox(0,0)[lb]{\smash{$\pprofile q {B \setminus A}$}}}%
    \put(0.03310004,0.3648626){\color[rgb]{0,0,0}\makebox(0,0)[lb]{\smash{$\pprofile q {B \cap A}$}}}%
    \put(0.73065938,0.29186332){\color[rgb]{0,0,0}\makebox(0,0)[lb]{\smash{$\begin{smallmatrix}
\langle 1,0,0\rangle,\\
\langle 0,1,0\rangle,\\
\langle 0,0,1\rangle~
\end{smallmatrix}$}}}%
    \put(0.42645544,0.49594119){\color[rgb]{0,0,0}\makebox(0,0)[lb]{\smash{$\begin{smallmatrix}
\langle 2,0,1\rangle, & \langle 1,0,2\rangle,\\
\langle 1,1,1\rangle, & \langle 0,1,2\rangle,\\
\langle 1,0,2\rangle, & \langle 0,0,3\rangle~
\end{smallmatrix}$}}}%
  \end{picture}%
\endgroup%